\documentclass[letter, 10pt, conference]{ieeeconf}
\IEEEoverridecommandlockouts        
\usepackage{graphics} 
\usepackage{verbatim}
\usepackage{epsfig} 
\usepackage{epstopdf}
\usepackage{amsmath} 
\usepackage{amssymb} 
\usepackage{geometry}
\usepackage[english]{babel}
\usepackage[utf8]{inputenc}
\usepackage{commath}

\usepackage{amsthm}
\usepackage{setspace}
\usepackage{amsfonts}
\let\labelindent\relax 
\usepackage{enumitem}
\usepackage{kotex}
\usepackage{color}
\usepackage[english]{babel}
\usepackage{algorithm2e}
\usepackage{algorithmic}
\usepackage{float}
\usepackage{cite}
\usepackage{textcomp}
\usepackage{subcaption}
\usepackage{mathrsfs}
\usepackage{cuted}
\usepackage{url}
\RestyleAlgo{ruled}

\DeclareMathAlphabet{\mymathbb}{U}{BOONDOX-ds}{m}{n}

\newtheorem{definition}{Definition}
\newtheorem{theorem}{Theorem}
\newtheorem{lemma}{Lemma}
\newtheorem{example}{Example}

\allowdisplaybreaks[3]

\title{\LARGE \bf
Data-Driven Inverse of Linear Systems \\ and Application to Disturbance Observers}
\author{Yongsoon Eun, Jaeho Lee, and Hyungbo Shim
\thanks{This work was supported by the National Research Foundation of Korea(NRF) grant funded by the Korea government(MSIT)(NRF2022R1A2C1010777).}%
\thanks{Yongsoon Eun and Jaeho Lee are with the Department of Electrical Engineering and Computer Science, DGIST, Daegu 42988, South Korea. {\tt\small \{yeun, jaeho.lee\} @ dgist.ac.kr}}%
\thanks{Hyungbo Shim is with ASRI, Department of Electrical and Computer Engineering, Seoul National University, Seoul, Korea. {\tt\small hshim@snu.ac.kr}}%
}

\begin{document}

\maketitle

\begin{abstract}
This work develops a data-based construction of inverse dynamics for LTI systems. 
Specifically, the problem addressed here is to find an input sequence from the corresponding output sequence based on pre-collected input and output data. 
The problem can be considered as a reverse of the recent use of behavioral approach, in which the output sequence is obtained for a given input sequence by solving an equation formed by pre-collected data. 
The condition under which the problem gives a solution is investigated and turns out to be $L$-delay invertibility of the plant and a certain degree of persistent excitation of the data input.
The result is applied to form a data-driven disturbance observer.
The plant dynamics augmented by the data-driven disturbance observer exhibits disturbance rejection without the model knowledge of the plant. 
\end{abstract}

\section{Introduction}\label{sec:introduction}
\noindent Seminal work of the behavioral approach \cite{Willems} for dynamic systems has been receiving growing attention in recent years \cite{Dorfler1,Coulson,Carlet}. Instead of using matrices $A$, $B$, $C$ and $D$ for linear dynamics, which are often obtained from the first principle on the target dynamics, 
the behavioral approach uses a set of equations to determine if a pair of input and output belongs to the target dynamics. The equations involve Hankel matrices for linear dynamics, which are obtained by input/output trajectories that satisfy a condition of persistently exciting. This approach suits when data are available but applying the first principle for modeling is not straightforward. 
Recent advances in this area include \cite{Waarde,Persis,Coulson,Dorfler2}.

Particularly, as given in \cite{MWH06}, the fundamental result of this approach is that we can determine a future output for a given input using Hankel matrices built by using previously collected input and output data. No identification of system matrices (typically $A$, $B$, $C$, $D$) is necessary. A question that this paper is concerned is the inverse of what is just described: {\em can we determine the input that generated a given output using  previously collected input and output data?} In other words, the question pertains to finding conditions and methods how we can build an inverse dynamics of the given system using collected data. If this is successfully carried out, an immediate application is to build a disturbance observer \cite{Ohishi,Shim} with data, that identifies the disturbance that affected the output. 
Another application is to find a suitable constrained input that yields an output close to the desired output under constraints.

Similar question has been asked and answered in \cite{Markovsky,MWH06}. However, the result is restricted to the case where $D$ is invertible. This work does not assume the invertibility of $D$. Instead the notion of $L$-delay invertibility is invoked from the literature \cite{Sain}. 

This paper is organized as follows: Section II provides a brief review of the invertibility of discrete-time LTI systems. Section III gives main results. An application of disturbance observer is given in Section IV. 
Section V concludes the paper.

\section{Review of LTI System Invertibility}
Consider a discrete-time linear-time-invariant system 
\begin{subequations} \label{sseq}
  \begin{align}
    x(k+1) &= Ax(k) + Bu(k) \label{stateupdate} \\
    y(k) &= Cx(k) + Du(k)
  \end{align}
\end{subequations}
where $x \in \mathbb{R}^n$ is the state, $u \in \mathbb{R}^m$ is the input, $y \in \mathbb{R}^p$ is the output, and $k \in \mathbb{Z}$ is the discrete time index.
Then the transfer function matrix of given system \eqref{sseq} is given by
$  G(z) = C(zI_n-A)^{-1}B + D$.

\begin{definition}[see \cite{Sain}] \label{delayedinverse}
A proper transfer function matrix $\hat{G}(z)$ 
is an {\em $L$-delay inverse system} for given system \eqref{sseq}, where $L$ is a nonnegative integer, if 
  \begin{align} \label{inversetf}
    \hat{G}(z)G(z) = \dfrac{1}{z^L}I_m.
  \end{align}
\end{definition} 

Note that $G(z)$ is an $p\times m$ transfer function matrix and $\hat G(z)$ is a $m \times p$ transfer function matrix. Thus, a necessary condition for \eqref{inversetf} is that $m\leq p$. 

\begin{definition}[see \cite{Sain}] \label{inherentdelay}
  The system \eqref{sseq} is {\em invertible} if it has an $L$-delay inverse for some finite $L$. The least integer $L$ for which an $L$-delay inverse exists will be called the {\em inherent delay} of the invertible system and is denoted by $L_0$.
\end{definition}

For single input and single output systems of \eqref{sseq}, the inherent delay $L_0$ is the relative degree. 

Define, for an integer $t>0$, 
\begin{align}\label{Toeplitz}
\begin{split}
    &\mathcal{T}_{t} \\
    &:= \begin{bmatrix}
    D         & \mymathbb{0} & \mymathbb{0} & \cdots & \mymathbb{0} \\
    CB        & D            & \mymathbb{0} & \cdots & \mymathbb{0} \\
    CAB       & CB           & D            & \cdots & \mymathbb{0} \\
    \vdots    & \vdots       & \vdots       & \ddots & \vdots \\
    CA^{t-1}B & CA^{t-2}B    & CA^{t-3}B    & \cdots & D
  \end{bmatrix}.
\end{split}
\end{align} 
In addition, define $\mathcal{T}_{0} = D$.
Then, the following theorem gives the  necessary and sufficient condition for the existence of an $L$-delay inverse for the system of \eqref{sseq}.

\begin{theorem}[see \cite{Sain}]\label{thm:1}
The following are equivalent:
\begin{enumerate}
\item The system \eqref{sseq} has an $L$-delay inverse.
\item There exists $L$ such that 
  \begin{align} \label{rankcondition}
  {\rm rank}(\mathcal{T}_{L}) - {\rm rank}(\mathcal{T}_{L-1}) = m,
  \end{align}
where ${\rm rank}(\mathcal{T}_{-1}) = 0$.
\item The rank of the matrix $G(z)$ is $m$.
\end{enumerate}
\end{theorem}

As a result of Theorem \ref{thm:1}, if the system has $L$-delay inverse, then the first $m$ columns of $\mathcal{T}_L$ are independent.

\section{Data-based Construction of Inverse Systems}

We introduce notations commonly used in the literature for data-driven system and control. 
For a given signal $f(k) \in \mathbb{R}^q$, $k \in \mathbb{Z}$, and $T \in \mathbb{N}$, let $f_{[k,k+T]}$ be defined by
$$f_{[k,k+T]} := \begin{bmatrix}
  f(k) \\
  f(k+1) \\
  \vdots \\
  f(k+T)
\end{bmatrix} \quad \in \mathbb{R}^{q(T+1)}.$$ 
Using the vectors $u_{[k,k+L]}$ and $y_{[k,k+L]}$, the system of \eqref{sseq} can be written as 
\begin{align} \label{iorelation}
  y_{[k,k+L]} = \mathcal{O}_{L} x(k) + \mathcal{T}_{L}u_{[k,k+L]}
\end{align}
where the matrix $\mathcal{O}_{L}$ is defined by 
\begin{align}\label{observability}
    \mathcal{O}_{L} = \begin{bmatrix}
    C^T \ (CA)^T  \  \cdots (CA^{L})^T
  \end{bmatrix}^T.
\end{align}
When the $L$-delay inverse exists for the system of \eqref{sseq}, there exists an $m\times p(L+1)$ matrix $\mathcal{K}$ satisfying 
$$\mathcal{K}\mathcal{T}_{L} = \begin{bmatrix}
  I_m & \mymathbb{0}_{m\times mL}
\end{bmatrix}.$$
Such a matrix $\mathcal{K}$ exists due to Rouch\'{e}-Capelli theorem \cite{RC} with the fact that
$${\rm rank}(\mathcal{T}_L) = {\rm rank}\left(\begin{bmatrix}
\mathcal{T}_L \\
  I_m \  \mymathbb{0}_{m\times mL}
\end{bmatrix}\right)$$
which follows from the fact that the first $m$ columns of $\mathcal{T}_L$ are independent.
Pre-multiplying $\mathcal{K}$ to both sides of \eqref{iorelation} yields 
\begin{align} \label{oirelation}
  u(k) = -\mathcal{K}\mathcal{O}_{L}x(k) + \mathcal{K}y_{[k,k+L]},  
\end{align}
which represents the input $u(k)$ by the state and sequence of outputs. 
State update equation with respect to $y_{[k,k+L]}$ is obtained by substituting \eqref{oirelation} in \eqref{stateupdate}. Then, the following dynamics is obtained:
\begin{subequations} \label{inversesseq}
  \begin{align}
    x(k+1) &= \tilde{A}x(k) + \tilde{B}y_{[k,k+L]} \\
    u(k) &= \tilde{C}x(k) + \tilde{D}y_{[k,k+L]}
  \end{align}
\end{subequations} where $\tilde{A} = A-B\mathcal{K}\mathcal{O}_L$, $\tilde{B} = B\mathcal{K}$, $\tilde{C} = -\mathcal{K}\mathcal{O}_{L}$, and $\tilde{D} = \mathcal{K}$.

In order to provide more specific arguments, we first remind readers of Hankel matrix and the notion of persistent excitation (PE) for a signal $f(k) \in \mathbb{R}^q$. Denote by $\mathcal{H}_t(f_{[i,j]})$ the $qt\times (j-i-t+2)$ dimensional Hankel matrix constructed using the values of $f(i), \cdots, f(j)$: 
\begin{align*}
    &\mathcal{H}_{t}(f_{[i,j]}) \\
    &= \begin{bmatrix}
      f(i) & f(i+1) & \cdots & f(j-t+1) \\
      f(i+1) & f(i+2) & \cdots & f(j-t+2) \\
      \vdots & \vdots & \ddots & \vdots \\
      f(i+t-1) & f(i+t) & \cdots & f(j) \\
    \end{bmatrix}.
\end{align*}
The notion of persistency of excitation of a segment from a signal is defined as follows.

\begin{definition}[see {\cite{Willems}}]
The vector $f_{[0,T-1]} \in \mathbb{R}^{qT}$ from signal $f(k) \in \mathbb{R}^{q}$ 
is {\em persistently exciting of order $t$} if the matrix $\mathcal{H}_t(f_{[0,T-1]}) \in \mathbb{R}^{qt \times (T-t+1)}$ has full row rank.
\end{definition}

As hinted by \eqref{inversesseq}, data-based construction of inverse dynamics may involve input and output sequences of different length. 

\begin{lemma}\label{lem:1}
A pair of vectors $u_{[0,T-1]} \in \mathbb{R}^{mT}$ and $y_{[0,T-1+L]} \in \mathbb{R}^{p(T+L)}$ is a trajectory of the system \eqref{sseq} if and only if there exists $u_{[T,T-1+L]}$ such that the pair 
$u_{[0,T-1+L]}$ and $y_{[0,T-1+L]}$ is a length $T+L$ trajectory of the system \eqref{sseq}. 
\end{lemma}

Now, for some $T>0$ and $L\ge 0$, denote the input and output data collected from an experiment for the system \eqref{sseq}
by $u^d_{[0,T+L-1]}$ and $y^d_{[0,T+L-1]}$, respectively. 

\begin{lemma}\label{lem:2}
 Let the system \eqref{sseq} be controllable. Assume that $u^d_{[0,T-1]}$ is persistently exciting of order $n+t+L$. 
 Then, for any $t_0$, a pair of a length $t$ input $u_{[t_0, t_0+t-1]}$ and a length $t+L$ output $y_{[t_0, t_0+t+L-1]}$ is a trajectory of \eqref{sseq} if and only if there is $g\in\mathbb{R}^{T-t+1}$ such that 
    \begin{align}\label{L2}
        \begin{bmatrix}
            u_{[t_0,t_0+t-1]} \\
            y_{[t_0,t_0+t+L-1]} 
        \end{bmatrix} = \begin{bmatrix}
            \mathcal{H}_{t}(u^d_{[0,T-1]}) \\
            \mathcal{H}_{t+L}(y^d_{[0,T+L-1]})
        \end{bmatrix}g.
    \end{align}
\end{lemma}

\begin{proof}
Let $u_{[t_0, t_0+t-1]}$ and $y_{[t_0,t_0+t+L-1]}$ be a trajectory of \eqref{sseq}. 
Then, by Lemma 1, there exists $u_{[t_0+t,t_0+t+L-1]}$ such that the pair $u_{[t_0, t_0+t+L-1]}$ and $y_{[t_0,t_0+t+L-1]}$ is a trajectory of \eqref{sseq}. 
Then, by \cite[Theorem 1]{Waarde} there exists $g$ such that
\begin{align}\label{BA}
    \begin{bmatrix}
        u_{[t_0,t_0+t+L-1]} \\
        y_{[t_0,t_0+t+L-1]} 
    \end{bmatrix} = \begin{bmatrix}
        \mathcal{H}_{t+L}(u^d_{[0,T+L-1]}) \\
        \mathcal{H}_{t+L}(y^d_{[0,T+L-1]}) 
    \end{bmatrix}g.
\end{align}
Hence \eqref{L2} is satisfied. 
Conversely, let $u_{[t_0, t_0+t-1]}$ and $y_{[t_0,t_0+t+L-1]}$ satisfy \eqref{L2} for some $g$. Set the subsequent sequence of $u$ as 
\begin{equation}
u_{[t_0+t,t_0+t+L-1]}= \mathcal{H}_{L}(u^d_{[t,T+L-1]})g.
\end{equation}  
Then, $u_{[t_0, t_0+t+L-1]}$ and $y_{[t_0,t_0+t+L-1]}$ satisfy \eqref{BA}. 
Again due to \cite[Theorem 1]{Waarde} the pair is the trajectory of \eqref{sseq}, and due to Lemma \ref{lem:1}, the pair $u_{[t_0, t_0+t-1]}$ and $y_{[t_0,t_0+t+L-1]}$ is a trajectory of \eqref{sseq}.
\end{proof}

Denote by $U_p$, $U_f$, $Y_p$ and $Y_{f_L}$ the following four Hankel matrices formed by the data $u^d_{[0,T+L-1]}$, $y^d_{[0,T+L-1]}$ for some $T_p>0$ and $T_f>0$ satisfying $T>T_p+T_f-1$  
\begin{align*}
    U_p &= \mathcal{H}_{T_p}(u^d_{[0,T-T_f-1]}),  \quad
    U_f = \mathcal{H}_{T_f}(u^d_{[T_p,T-1]}), \\ 
    Y_p &= \mathcal{H}_{T_p}(y^d_{[0,T-T_f-1]}),  \quad
    Y_{f_L} = \mathcal{H}_{T_f+L}(y^d_{[T_p,T+L-1]}).
\end{align*}

\begin{lemma}\label{lem:3}
Let  the system of \eqref{sseq} be observable and have an $L$-delay inverse with $L\ge0$. Assume that $T_p$ be greater than or equal to the observability index of \eqref{sseq}. 
Then,
\begin{equation}\label{eq:equalnull}
\mathcal{N}\left(\begin{bmatrix}
      U_p \\ 
      Y_p \\ 
      U_f \\ 
      Y_{f_L} 
    \end{bmatrix}\right) = \mathcal{N}\left(\begin{bmatrix}
      U_p \\ 
      Y_p \\ 
      Y_{f_L} 
    \end{bmatrix}\right)
\end{equation}
where $\mathcal{N}(A)$ is the null space of matrix $A$.
\end{lemma}

\begin{proof}
Denote $[x(T_p), x(T_p+1), \cdots, x(T-T_f)]$ by $X$, and define $v(k)=y^d_{[k,k+L]}$ for the simplicity of derivation. 
First, we show that $U_f$ is determined by $X$ and $\mathcal{H}_{T_f}(v_{[T_p,T-1]})$ by \eqref{oirelation} and \eqref{inversesseq} as
\begin{align*}
    U_f &= \begin{bmatrix}
      \tilde{C} \\
      \tilde{C} \tilde{A} \\
      \vdots \\
      \tilde{C} \tilde{A}^{T_f -1}
    \end{bmatrix} X \\
    &+ \begin{bmatrix}
      \tilde{D} & \mymathbb{0}_{m\times p(L+1)} & \cdots & \mymathbb{0}_{m\times p(L+1)} \\
      \tilde{C} \tilde{B} & \tilde{D} & \cdots & \mymathbb{0}_{m\times p(L+1)} \\
      \vdots & \ddots & \ddots & \vdots \\
      \tilde{C} \tilde{A}^{T_f - 2} \tilde{B} & \tilde{C} \tilde{A}^{T_f - 3} \tilde{B} & \cdots & \tilde{D}
    \end{bmatrix} \\
    &\quad \times \mathcal{H}_{T_f}(v_{[T_p,T-1]}).
\end{align*}
Again using \eqref{Toeplitz}, \eqref{observability}, and \eqref{inversesseq}, $X$ is determined by $U_p$ and $Y_p$ as 
\begin{align}\label{X}
  X = A^{T_p}\mathcal{O}_{T_p - 1}^\dagger(Y_p-\mathcal{T}_{T_p - 1}U_p) + \mathcal{C}_{T_p-1}U_p
\end{align} 
where 
\begin{align*}
\mathcal{C}_{T_p-1} = \begin{bmatrix}
  A^{T_p-1}B & A^{T_p-2}B & \cdots & AB & B
\end{bmatrix}
\end{align*}
and $\dagger$ implies the left-inverse of a matrix.
Finally, let $M_i \in \mathbb{R}^{p(L+1)\times p(T_f+L)}$ be a matrix that include the identity matrix of size $p(L+1)$ from its $p(i-1)+1$-th column and all other elements are zero.
Then, it is easily seen that 
\begin{align}\label{Hv}
    \mathcal{H}_{T_f}(v_{[T_p,T-1]})= M Y_{f_L}
\end{align}    
where $M^T = [M_1^T, \cdots, M_{T_f}^T]$.
Therefore, $\mathcal{H}_{T_f}(v_{[T_p,T-1]})$ is determined by $Y_{f_L}$.
Combining \eqref{X} and \eqref{Hv}, we arrive that $U_f$ linearly depends on $U_p$, $Y_p$ and $Y_{f_L}$. 
Therefore, \eqref{eq:equalnull} follows.
\end{proof}
Now we state the main result of this paper. 
\begin{theorem}\label{thm:2}
Let the system of \eqref{sseq} be controllable, observable and have an $L$-delay inverse with $L \ge 0$. 
Assume that $T_p$ be greater than or equal to the observability index of \eqref{sseq} and $u^d_{[0,T-1]}$ is PE of order $n+T_p+T_f+L$. 
Then, for any trajectory $u_{[t_0,t_0+T_p-1]}$ and $y_{[t_0,t_0+T_p+T_f+L-1]}$ of \eqref{sseq} where $t_0$ is arbitrary, the input sequence  $u_{[t_0+T_p,t_0+T_p+T_f-1]}$ is uniquely determined by
\begin{equation}\label{eq:found}
    u_{[t_0+T_p,t_0+T_p+T_f-1]}= U_f g,
\end{equation}
where $g$ is a solution of 
\begin{equation}\label{ExtendedHankel}
    \begin{bmatrix}
      U_p \\
      Y_p \\ 
      Y_{f_L}
    \end{bmatrix} g = \begin{bmatrix}
      u_{[t_0,t_0+T_p-1]} \\ 
      y_{[t_0,t_0+T_p-1]} \\ 
      y_{[t_0+T_p,t_0+T_p+T_f+L-1]} 
    \end{bmatrix}.
\end{equation}
\end{theorem}

\begin{proof}
We first claim that \eqref{ExtendedHankel} always has a solution. 
Since $u_{[t_0,t_0+T_p-1]}$ and $y_{[t_0,t_0+T_p+T_f+L-1]}$ are (different length) trajectory of \eqref{sseq}, there exists a $T_f$ long sequence $u^*$ that satisfies
\begin{equation}
    \begin{bmatrix}
      U_p \\
      Y_p \\ 
      U_f \\
      Y_{f_L}
    \end{bmatrix} g = \begin{bmatrix}
      u_{[t_0,t_0+T_p-1]} \\ 
      y_{[t_0,t_0+T_p-1]} \\ 
      u^* \\
      y_{[t_0+T_p,t_0+T_p+T_f+L-1]} 
    \end{bmatrix}
\end{equation}
by Lemmas \ref{lem:1} and \ref{lem:2}, which proves the first claim.

We note that, for any solution $g$ to \eqref{ExtendedHankel}, the vector $U_f g$ is unique.
Indeed, if $g_1$ and $g_2$ are two solutions of \eqref{ExtendedHankel}, then $g_1-g_2 \in \mathcal{N}([U_p^T,Y_p^T,Y_{f_L}^T]^T)$.
This in turn implies that $g_1 - g_2 \in \mathcal{N}(U_f)$ by \eqref{eq:equalnull}, so that $U_f(g_1-g_2)=0$.

Therefore, set $u_{[t_0+T_p, t_0+T_p+T_f-1]} = U_fg$ for any solution $g$ of \eqref{ExtendedHankel}. 
Then, it holds that 
\begin{equation}
    \begin{bmatrix}
      U_p \\
      Y_p \\ 
      U_f \\
      Y_{f_L}
    \end{bmatrix} g = \begin{bmatrix}
      u_{[t_0,t_0+T_p-1]} \\ 
      y_{[t_0,t_0+T_p-1]} \\ 
      u_{[t_0+T_p, t_0+T_p+T_f-1]} \\
      y_{[t_0+T_p,t_0+T_p+T_f+L-1]} 
    \end{bmatrix},
\end{equation}
which implies that the pair $u_{[t_0, t_0+T_p+T_f-1]}$ and $y_{[t_0, t_0+T_0+T_f+L-1]}$ are trajectories of \eqref{sseq} by Lemma~\ref{lem:2}. 
\end{proof}

\noindent Theorem \ref{thm:2} provides a foundation to answer the question posed in Introduction, i.e.,  determining the input that generated a given output using previously collected data, since it states that 
$T_f$ long input $u_{[t_0+T_p,t_0+T_p+T_f-1]}$ is uniquely determined for any pair $u_{[t_0,t_0+T_p-1]}$ and $y_{[t_0,t_0+T_p+T_f+L-1]}$. Figure \ref{fig:estimation} illustrates this with color coded signals when $T_p=4$, $T_f=3$ and $L=2$. Signals in black are required to determine the signal in red. 

\begin{figure}[!ht]
  \centering 
  \includegraphics[width=0.8\linewidth]{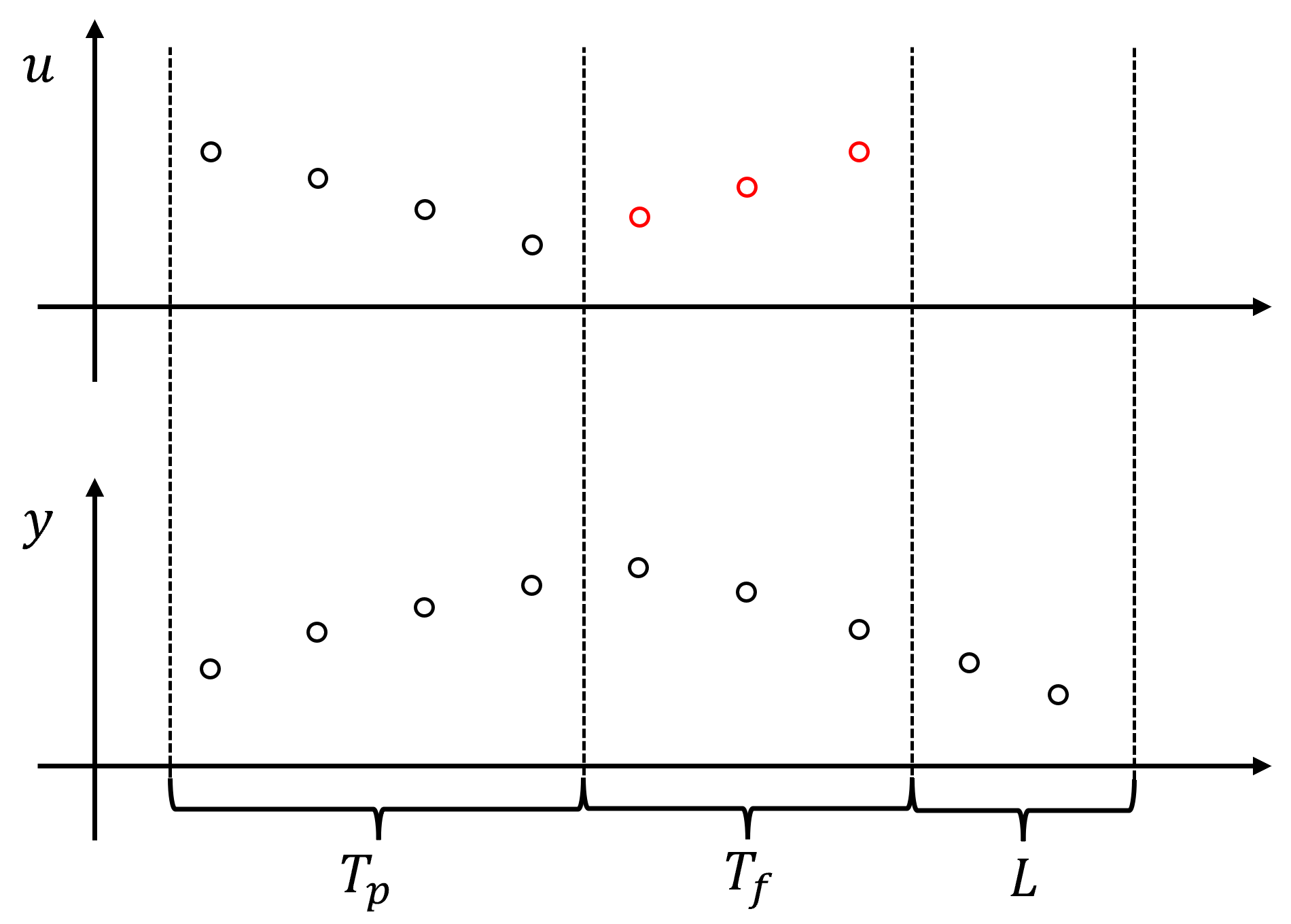}
  \caption{Graphical illustration of sequences in Theorem \ref{thm:2}}
  \label{fig:estimation}
\end{figure}

The necessity of $T_p$ long pair of previous $u$ and $y$ is for the system state $x(t_0+T_p)$ identification. From this perspective, an interpretation of the theorem is given as follows. Any vector $g$ belonging to
$$\mathcal{G} := \left\{ g : \begin{bmatrix} U_p \\ Y_p \end{bmatrix} g = \begin{bmatrix} u_{[t_0,t_0+T_p-1]} \\ y_{[t_0,t_0+T_p-1]} \end{bmatrix} \right\}$$
contains the information of the state $x(t_0+T_p)$, so that the set
$$\left\{ \begin{bmatrix} U_f \\ Y_{f_L} \end{bmatrix} g : g \in \mathcal{G} \right\}$$
is a collection of all possible $T_f$-long input and $(T_f+L)$-long output sequences starting from $x(t_0+T_p)$.
Thus, any observed $y_{[t_0+T_p,t_0+T_p+T_f+L-1]}$  must belong to this set, i.e., there exists $g\in\mathcal{G}$ and 
$$ y_{[t_0+T_p,t_0+T_p+T_f+L-1]} = Y_{f_L} g. $$
This $g$ is identified by solving \eqref{ExtendedHankel}. The corresponding  $u_{[t_0+T_p,t_0+T_p+T_f-1]}$ is obviously given as in \eqref{eq:found} and unique.

The result of Theorem \ref{thm:2} is written in the form of algorithm. See Algorithm 1. 
\begin{algorithm}[ht] 
  \begin{algorithmic}[1] 
    \caption{Input Estimation}
    \STATE {\bf{Setup: }} Set $T_p, T_f, L \in \mathbb{Z}$ with $T_p \geq $ (observability index of the system). 
    Make Hankel matrices $U_p$, $U_f$, $Y_p$, $Y_{f_L}$ from persistently exciting input of order $n+T_p+T_f+L$.
    Get $u_{[-T_p,-1]}$ and $y_{[-T_p,-1]}$.
    Set $k=0$.
    \STATE {\bf{Input: }} $y_{[k,k+T_f+L-1]}$ 
    \STATE Solve $g$ for 
    \begin{align*}
        \begin{bmatrix}
          Y_p \\
          U_p \\
          Y_{f_L} 
        \end{bmatrix}g = \begin{bmatrix}
          y_{[k-T_p,k-1]} \\
          u_{[k-T_p,k-1]} \\
          y_{[k,k+T_f+L-1]}
        \end{bmatrix}.
    \end{align*}
    \STATE Set $u_{[k,k+T_f-1]} = U_f g$.
    \STATE {\bf{Output: }} $\hat u = u_{[k,k+T_f-1]}$
    \STATE Update $k \leftarrow k+T_f$
    \STATE Repeat from 2
\end{algorithmic}
\end{algorithm}

To initiate Algorithm 1, $u_{[-T_p,-1]}$ and $y_{[-T_p,-1]}$ must be obtained. If the system is initially at rest, both can be set to zeros. An example is provided. 

\begin{example}\label{exam:1}
As an example, we consider the system of \eqref{sseq} with 
\begin{subequations}\label{example1}
\begin{align}
    A &= \begin{bmatrix}
        -0.3 & 0 \\
        0 & -0.5
    \end{bmatrix}, \quad B=\begin{bmatrix}
        2 & 1 \\
        -1 & 1
    \end{bmatrix}, \\
    C&= \begin{bmatrix}
        1 & 2 \\
        1 & 0
    \end{bmatrix}, \quad D= \begin{bmatrix}
        0 & 0 \\
        0 & 0
    \end{bmatrix}.
\end{align}
\end{subequations}
The system has $L$-inverse with $L=1$. Applying  $u_{[0,99]}$ as shown in Figure \ref{fig:eg1_io}(a) to the system with $x(0)=0$ yields  $y_{[0,99]}$ as shown in Figure \ref{fig:eg1_io}(b). 

\begin{figure}[!ht]
     \centering
     \begin{subfigure}[ht]{0.45\linewidth}
         \centering
         \includegraphics[width=\linewidth]{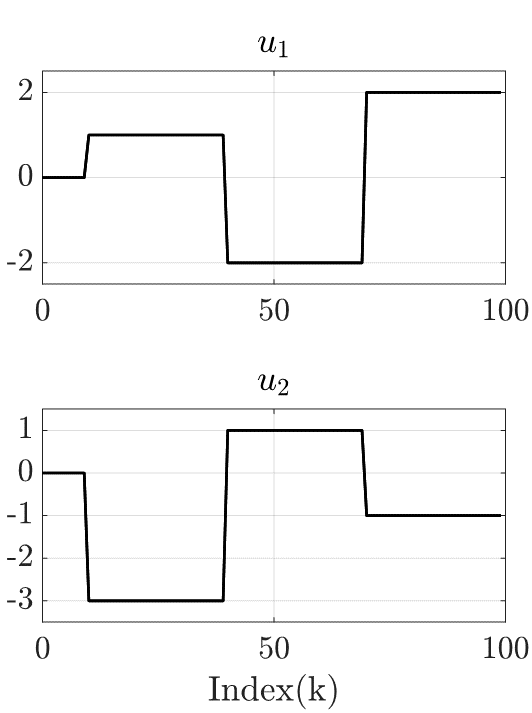}
         \caption{Inputs}
         \label{fig:eg1_input}
     \end{subfigure} \hfill
     \begin{subfigure}[ht]{0.45\linewidth}
         \centering
         \includegraphics[width=\linewidth]{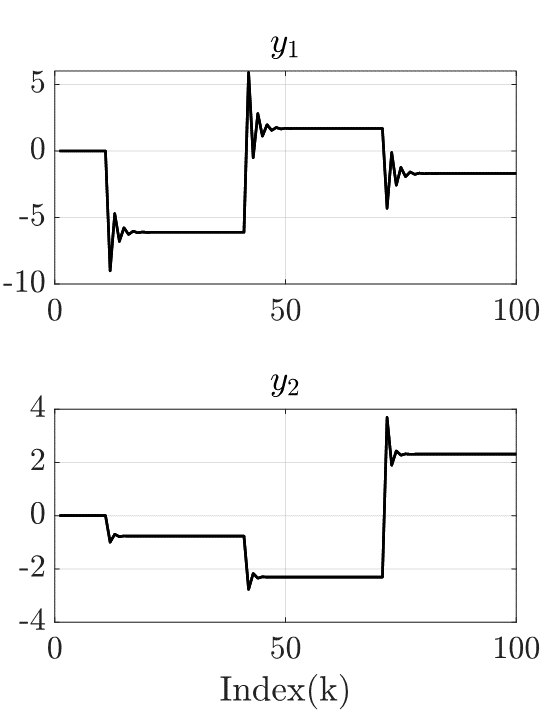}
         \caption{Outputs}
         \label{fig:eg1_output}
     \end{subfigure}
     \caption{Inputs and outputs}
     \label{fig:eg1_io}
\end{figure}

To apply Algorithm 1, we set $T_p=2$, $T_f=3$, and $L=1$, and make Hankel matrices with data generated by a 30 long random input sequence satisfying PE of order $n+T_p+T_f+L = 8$. 
Assuming the system initially at rest, Algorithm 1 begins iteration at time $k=3$ with $u_{[0,1]}=0_{4 \times 1}$, $y_{[0,1]}=0_{4 \times 1}$ and the first input to the algorithm is $y_{[2,5]}$.
Denoting the outcome of Algorithm by $\hat u(k)$ at time $k$, we obtain $\hat u_{[2,97]}$, which is shown in Figure \ref{fig:eg1_input_est} in red. 
\begin{figure}[!ht]
  \centering 
  \includegraphics[width=0.9\linewidth]{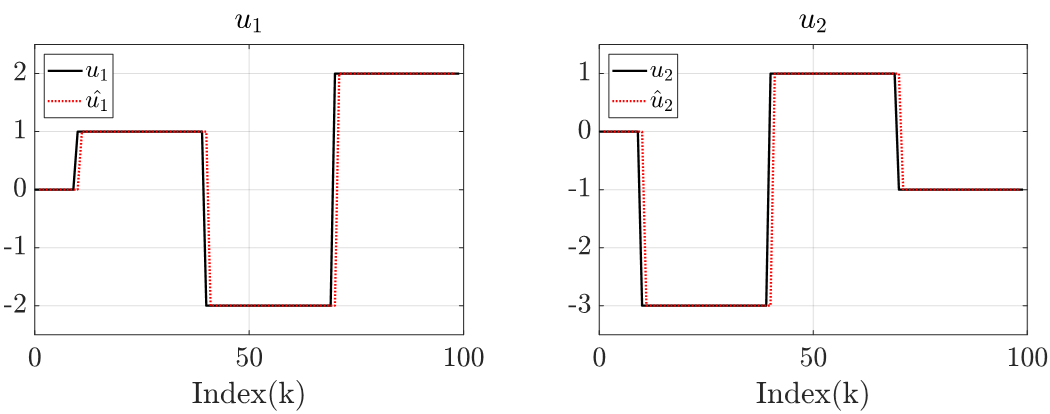}
  \caption{Estimation of inputs}
  \label{fig:eg1_input_est}
\end{figure}
Clear, $\hat u$ and $\hat u$ are identical with $L=1$ step delay. 

Note that $\hat u$ is computed in batches of length $T_f=3$. Iteration $k$ gives $\hat u_{[3k,3k+2]}$.
\end{example}

Algorithm 1 receives $y_{[k,k+T_f+L-1]}$ as an input and  computes $\hat u = u_{[k, k+T_f-1]}$ at a time. This may be cumbersome to implement for real-time applications. We provide Algorithm 2 tailored for real-time application. Here, the input is modified to use $y$ available up to the current instance and compute $\hat u$ of length 1. 

\begin{algorithm}[ht] 
  \begin{algorithmic}[1] 
    \caption{Input Estimation (with $T_f=1$)}
    \STATE {\bf{Setup: }} Set $T_p, L \in \mathbb{Z}$ with $T_p \geq $ (observability index of the system). 
    Make Hankel matrices $U_p$, $U_f$, $Y_p$, $Y_{f_L}$ from persistently exciting input of order $n+T_p+1+L$.
    Get $u_{[-T_p-L,-L-1]}$ and $y_{[-T_p-L,-1]}$.
    Set $k=0$.
    \STATE {\bf{Input: }} $y(k)$ 
    \STATE Solve $g$ for 
    \begin{align*}
        \begin{bmatrix}
          Y_p \\
          U_p \\
          Y_{f_L} 
        \end{bmatrix}g = \begin{bmatrix}
          y_{[k-T_p-L,k-L-1]} \\
          u_{[k-T_p-L,k-L-1]} \\
          y_{[k-L,k]}
        \end{bmatrix}.
    \end{align*}
    \STATE Set $u(k-L) = U_f g$.
    \STATE {\bf{Output: }} $\hat u(k) = u(k-L)$
    \STATE Update $k \leftarrow k+1$
    \STATE Repeat from 2
\end{algorithmic}
\end{algorithm}

\begin{example}
Consider the system of \eqref{sseq} with 
\begin{subequations}\label{example2}
\begin{align}
    A&=  \begin{bmatrix}
        1    & 0.05 & 0    & 0.1  & 0.5  & 0    \\
        0.05 & 1    & 0.05 & 0.05 & 0.1  & 0.05 \\
        0    & 0.05 & 1    & 0    & 0.05 & 0.1  \\
        -0.2 & 0.1  & 0.05 & 0.8  & 0.1  & 0.05 \\
        0.1  & -0.2 & 0.1  & 0.1  & 0.8  & 0.1  \\
        0    & 0.1  & -0.2 & 0.05 & 0.1  & 0.8
    \end{bmatrix},\nonumber     \\
    B&=\begin{bmatrix}
        0.1  & 0    \\
        0.1  & 0    \\
        0    & 0.1  \\
        0.1  & 0.05 \\
        -0.1 & 0.1  \\
        0    & - 0.1
    \end{bmatrix}, \quad  
    C^T = \begin{bmatrix}
        1 & 0 & 0 \\
        0 & 1 & 0 \\
        0 & 0 & 1 \\
        0 & 0 & 0 \\
        0 & 0 & 0 \\
        0 & 0 & 0 
    \end{bmatrix}, \nonumber
\end{align}
\end{subequations}
and $D=0_{3 \times 2}$. This system has $L_0$-delay inverse with $L_0=1$.

Algorithm 2 is used with $L=1$ and $T_p=2$. Figure \ref{fig:eg2_io} shows the input and output sequence. Figure \ref{fig:eg2_inputs_est} shows the $\hat u$ along with $u$. Indeed they are identical with one step delay. 
\begin{figure}[!ht] 
     \centering
     \begin{subfigure}[ht]{0.42\linewidth}
         \centering
         \includegraphics[width=\linewidth]{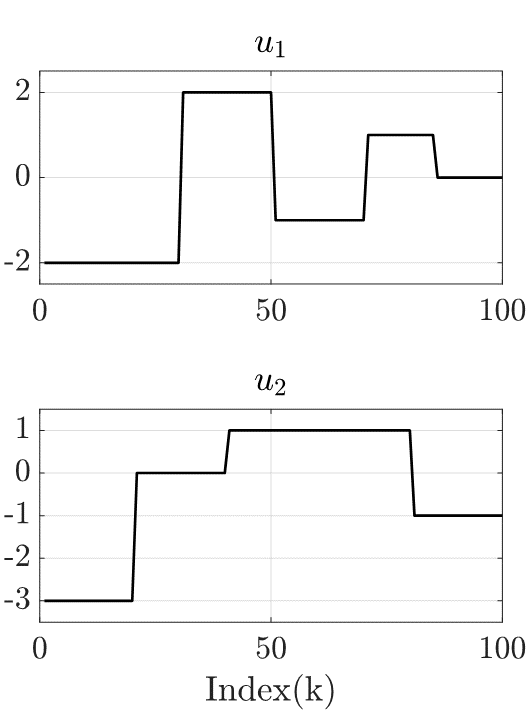}
         \caption{Inputs}
         \label{fig:eg2_input}
     \end{subfigure} \hfill
     \begin{subfigure}[ht]{0.40\linewidth}
         \centering
         \includegraphics[width=\linewidth]{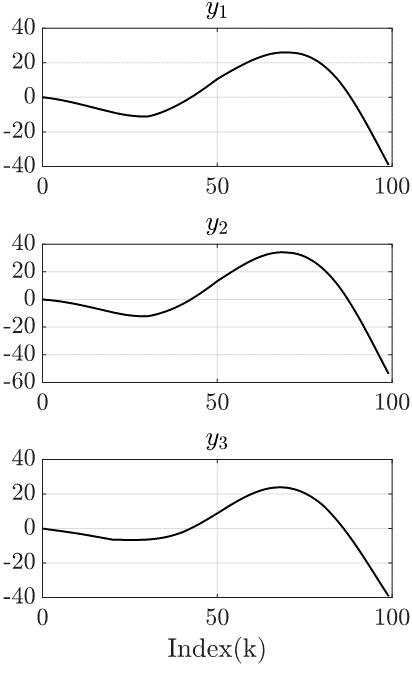}
         \caption{Outputs}
         \label{fig:eg2_output}
     \end{subfigure}
     \caption{Inputs and outputs}
     \label{fig:eg2_io}
\end{figure}

\begin{figure}[!ht]
  \centering 
  \includegraphics[width=0.9\linewidth]{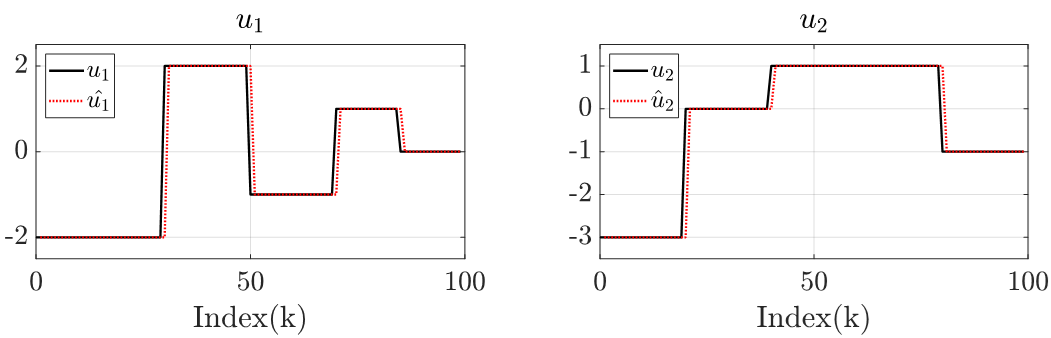}
  \caption{Estimation of inputs}
  \label{fig:eg2_inputs_est}
\end{figure}
\end{example}

\begin{example}\label{exam:3} 
The proposed inversion of a linear system can also be used for finding a constrained input that yields a desired output. 
Consider a desired output $y^*_{[0,T_f+L-1]}$ and an input constraint set $\mathcal{U} \subset \mathbb{R}^m$.
Suppose that we want to find a constrained input sequence $u_{[0,T_f-1]}$ at time $k=0$ when the previous records of input $u_{[-T_p,-1]}$ and output $y_{[-T_p,-1]}$ are available.
Then, solve the minimization problem for $g$:
\begin{align*}
\min_{g} \; & \|Y_{f_L} g - y^*_{[0,T_f+L-1]}\|^2 \\
\text{subject to} \; & \begin{bmatrix} U_p \\ Y_p \end{bmatrix} g = \begin{bmatrix} u_{[-T_p,-1]} \\ y_{[-T_p,-1]} \end{bmatrix} \\
& U_f g \in \mathcal{U} \times \cdots \times \mathcal{U}
\end{align*}
and set $u_{[0,T_f-1]} = U_f g$.
\end{example}

\section{Application to Disturbance Observer}\label{sec:app}
Consider the block diagram of Figure \ref{fig:DDOB}, where $G(z)$ is the plant and the block `{\tt DBINV}' is the data-based inversion block of $G(z)$ implemented by  Algorithm 2.
The signals $u_0$, $y$, $d$, $u$ are, respectively, command input, output, disturbance, and the actual input to the plant. 
Note that the transfer function from $u_0$ to $y$ is  
\begin{equation}\label{DDOBTF}
    Y(z) = G(z)[U_0(z) + (1-z^{-L})D(z)]
\end{equation}
where $Y$, $U_0$, and $D$ are $z$-transform of the corresponding signals, which can also be seen from the figure by
\begin{gather*}
\hat u(k) = d(k-L) + \Delta(k-L), \quad \hat d(k) = d(k-L), \\
\Delta(k) = u_0(k) - \hat d(k) = u_0(k) - d(k-L).
\end{gather*}
From \eqref{DDOBTF} it is seen that, when $d(k)$ is slowly varying, the effect of $d$ is approximately removed from the input. 
This is exactly the idea of the model-based disturbance observer \cite{Shim}.
In this sense, we refer the proposed structure as {\it data-driven disturbance observer}.

\begin{figure}[!ht]
  \centering 
  \includegraphics[width=\linewidth]{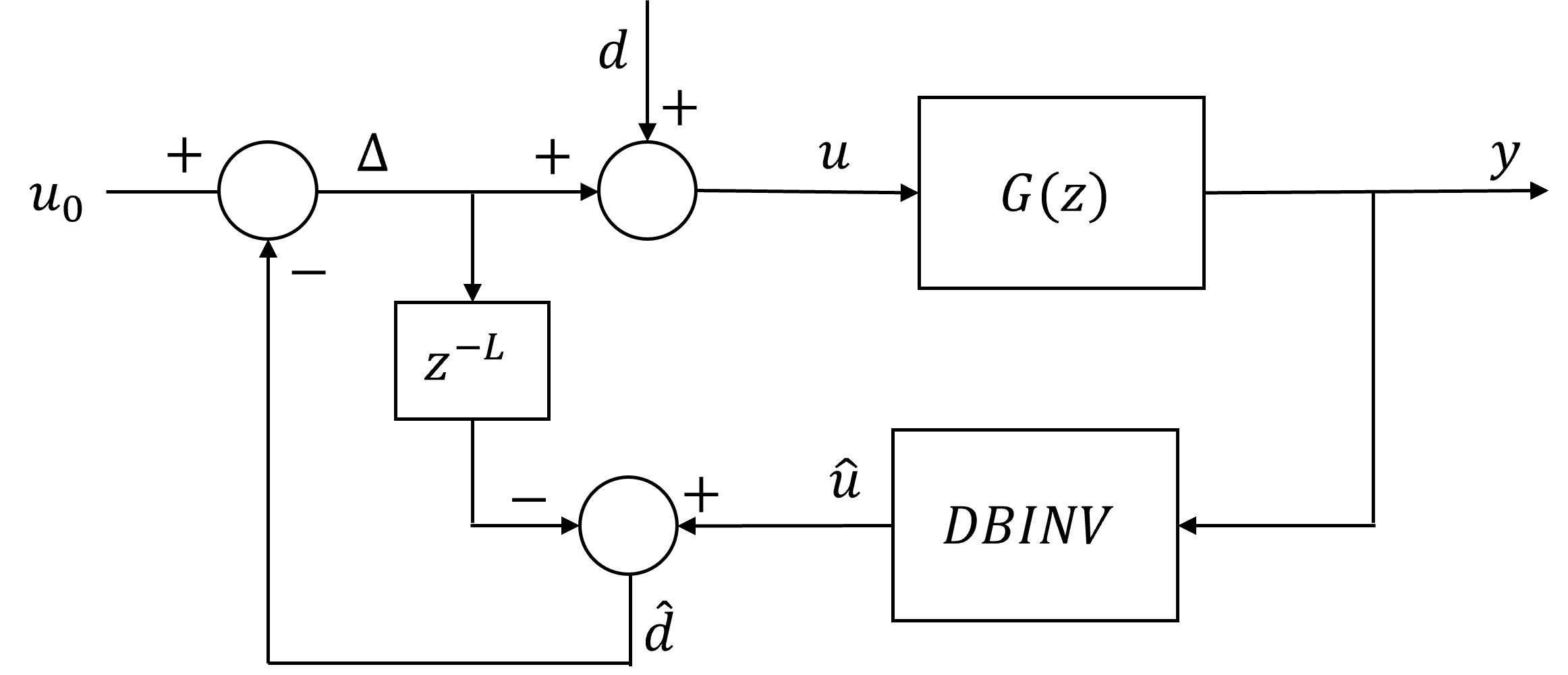}
  \caption{Data-driven Disturbance Observer}
  \label{fig:DDOB}
\end{figure}

\begin{example}
Consider the block diagram of Figure \ref{fig:DDOB} where a realization of $G(z)$ is given by the system of \eqref{example1}. 
We assume that the signal $u_0$ is the same as the signal shown in Figure \ref{fig:eg1_io}(a) and the signal $d$ is given in Figure \ref{fig:eg2_DDOB}(a). 
The corresponding output is shown in Figure \ref{fig:eg2_DDOB}(b). 
Clearly, the effect of disturbance is removed after some transients.

\begin{figure}[!ht]
     \centering
     \begin{subfigure}[ht]{0.45\linewidth}
         \centering
         \includegraphics[width=\linewidth]{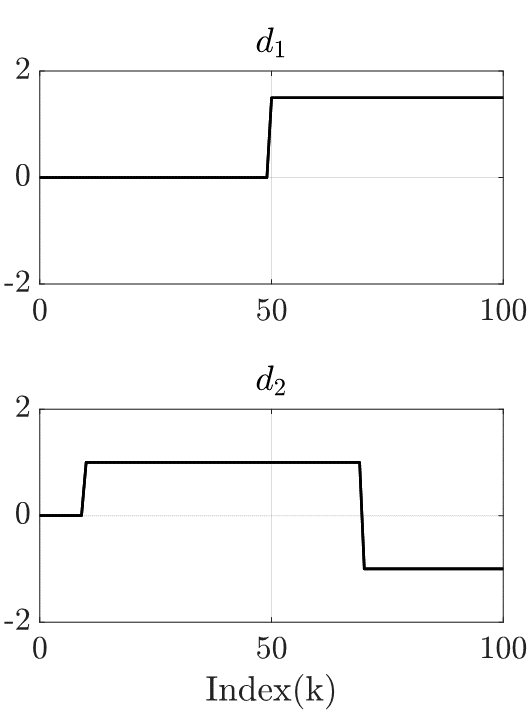}
         \caption{Disturbances}
         \label{fig:eg2_disturbances2}
     \end{subfigure} \hfill
     \begin{subfigure}[ht]{0.45\linewidth}
         \centering
         \includegraphics[width=\linewidth]{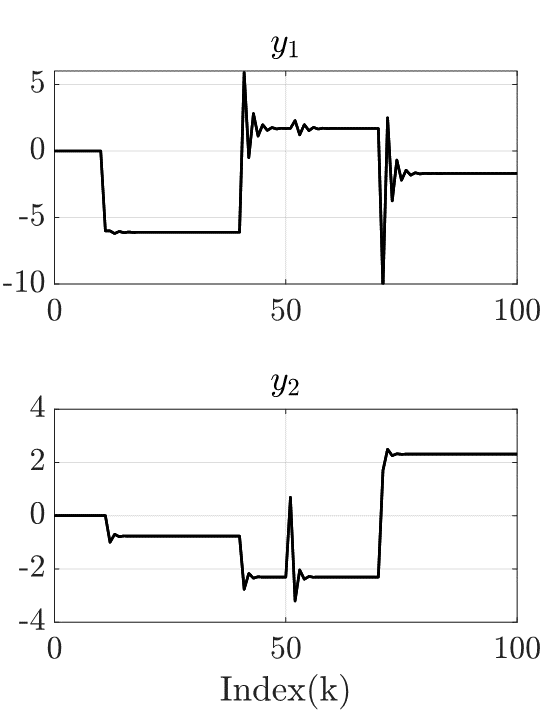}
         \caption{Outputs}
         \label{fig:eg2_outputs2}
     \end{subfigure}
     \caption{Disturbances and outputs with disturbance observer}
     \label{fig:eg2_DDOB}
\end{figure}

The simulation results agree with the derivation in \eqref{DDOBTF} and we emphasize that no information of $G(z)$ is directly used, and thus, the term of \emph{data-driven disturbance observer} is justified.

\end{example}

\section{Conclusions}
A data-based construction of inverse dynamics has been developed for discrete-time LTI systems. 
The notion of $L$-delay inverse is invoked from the literature, and it is combined with the system description method from behavioral approach. 
The outcome is data-based representation of inverse dynamics, which is similar to that of LTI systems, but differs in that input is estimated with some delay. 
The result is applied to build disturbance observers from no system information but collected input and output data that satisfy a level of persistency of excitation condition. 
Applications and extension this result seem to be immense.

\end{document}